\documentclass[11pt]{article}

\usepackage[margin=1in]{geometry}
\usepackage{amsmath}
\usepackage{amssymb}
\usepackage{amsthm}
\usepackage{tabularx}
\usepackage{booktabs}
\usepackage{array}
\usepackage{float}
\usepackage{tikz}
\usepackage{pgfplots}
\pgfplotsset{compat=1.18}
\usepackage[colorlinks=true,citecolor=blue,linkcolor=blue,urlcolor=blue]{hyperref}

\newtheorem{theorem}{Theorem}[section]

\newtheorem{corollary}[theorem]{Corollary}
\newtheorem{proposition}[theorem]{Proposition}
\newtheorem{definition}[theorem]{Definition}
\newtheorem{assumption}[theorem]{Assumption}
\theoremstyle{remark}
\newtheorem{remark}{Remark}[section]

\begin{document}
	
	\title{Scarcity Is Not Enough: Structural Limits of
		Linear Sybil Cost Under Parallelizable Resources}
	
	\author{
		Homayoun Maleki\thanks{Corresponding author: \texttt{h.maleki@deusto.es}} \\
		DeustoTech, University of Deusto, Bilbao, Spain
		\and
		Nekane Sainz \\
		DeustoTech, University of Deusto, Bilbao, Spain
		\and
		Jon Legarda \\
		DeustoTech, University of Deusto, Bilbao, Spain
		\and
		Igor Santos-Grueiro \\
		International University of La Rioja (UNIR), Logro\~{n}o, Spain
	}
	
	\date{}
	
	\maketitle
	
	\begin{abstract}
		Permissionless systems resist Sybil attacks by binding influence
		to scarce resources. Yet influence concentration persists across
		systems built on computation, capital, and other reusable resources
		despite substantial differences in protocol design. This raises a
		fundamental question: is influence concentration primarily a
		consequence of protocol rules, or of the structural properties of
		the underlying resource itself?
		
		In this paper, we address this question from a resource-centric
		perspective through the adversarial cost function $C(s,T)$,
		defined as the minimum expenditure required to sustain
		influence equivalent to controlling $s$ independent
		participants over a time horizon of length $T$. Using this framework, we develop an axiomatic resource
		taxonomy that connects resource structure to adversarial cost
		scaling.
		
		Our analysis shows that resource-mechanism pairs satisfying
		divisibility, additivity of influence, temporal reusability, and
		identity transferability admit influence amortization, yielding
		$C(s,T)=o(sT)$. Conversely, throughput-bounded,
		non-transferable, window-local resources enforce
		$C(s,T)=\Omega(sT)$, with marginal cost
		$\Delta(s,T)=\Omega(T)$ increasing with the time horizon.
		Together, these results reveal a fundamental asymptotic
		separation between resource classes that admit influence
		amortization and those that enforce linear cost.
		
		These results shift the focus of decentralization from protocol
		design to resource design. If influence concentration is a
		structural consequence of the properties that make resources
		parallelizable, then redesigning consensus rules alone cannot
		eliminate it. The same properties also enable concentrated
		control to be projected across many nominally distinct
		participants through delegation, pooling, or identity
		replication, obscuring the relationship between visible
		identities and underlying control. The search for stronger
		decentralization must therefore begin with the choice of the
		underlying resource.
		
	\end{abstract}

	
	\section{Introduction}
	\label{sec:intro}
	
	Permissionless systems rest on a radical premise: that open
	participation, without gatekeepers or trusted authorities,
	can produce secure and fair collective outcomes.
	The mechanism that makes this possible---binding influence
	to the expenditure of a scarce resource---has been the
	cornerstone of distributed consensus since
	Nakamoto~\cite{nakamoto2008bitcoin}.
	Yet deployment has revealed a persistent gap between
	this promise and reality.
	Bitcoin's hash power is concentrated in a handful of
	mining pools~\cite{gencer2018decentralization}.
	Ethereum's stake is dominated by a small number of
	liquid staking providers~\cite{lido2023concentration}.
	
	This paper identifies a structural explanation for this phenomenon.
	
	\subsection{Resource-Based Defenses and Their Shared Blind Spot}
	\label{ssec:sybil}
	
	The foundational obstacle in permissionless systems is the
	Sybil attack~\cite{douceur2002sybil}: a single adversary
	generates many identities to amplify influence beyond any
	honest participant.
	The standard response is to make identity costly by tying
	participation to a scarce resource.
	Proof-of-Work binds influence to computational
	effort~\cite{nakamoto2008bitcoin,garay2015backbone}
	(specifically, to mining hardware capacity as a
	reusable stock, and energy as a per-window flow).
	Proof-of-Stake binds it to financial
	capital~\cite{kiayias2017ouroboros}.
	Graph-based defenses~\cite{yu2006sybilguard,yu2008sybillimit}
	constrain identity placement via social structure.
	Identity-binding mechanisms~\cite{borge2017pop,ford2020pseudonym}
	restrict credential issuance to verified participants.
	
	Despite their diversity, these approaches share a common
	blind spot.
	The original insight behind resource-weighted
	participation---introduced by Proof-of-Work---was
	precise: binding influence to a scarce resource
	prevents cheap identity creation from amplifying
	adversarial power.
	Subsequent work has studied this principle extensively,
	asking whether an adversary's \emph{aggregate} resource
	share crosses a safety threshold.
	
	What has received far less attention is a different
	question: does the \emph{structure} of that resource
	permit influence to be concentrated at sublinear cost?
	A resource that is divisible, transferable, and reusable
	allows a stock of size $O(s)$ to sustain $s$-unit
	influence over an arbitrarily long horizon without any
	$T$-proportional renewal cost. The critical issue is
	therefore not merely scarcity, but whether influence can
	be amortized over time.
	
	More importantly, the same resource properties allow
	concentrated control to be projected across many
	nominally distinct participants. A mining pool, for
	example, may aggregate a large concentration of hash
	power while distributing block-production opportunities
	across thousands of visible miners. Similar effects arise
	through delegation and other coordination mechanisms,
	where a single underlying concentration of resources can
	appear as a diverse population of independent actors.
	The challenge is therefore not only influence
	concentration, but the growing disconnect between
	visible participation and effective control.

	This paper characterizes exactly which resource properties
	permit this temporal amortization, and which ones prevent it.
	
	\subsection{Why Concentration Emerges as a Structural Consequence}
	\label{ssec:concentration}
	
	The concentration patterns observed across deployed systems
	are often discussed as protocol-level phenomena.
	Our framework suggests that they may instead arise from the
	structural properties of the resources on which those
	protocols are built.
	
	In Bitcoin, mining pools aggregate hash power from thousands
	of participants into a small number of coordinated entities.
	Gencer et al~\cite{gencer2018decentralization} document
	that the top-four pools routinely control more than half of
	the network's total hash power.
	The reason may be structural rather than purely procedural.
	Mining hardware capacity---the reusable stock component of
	Proof-of-Work---can be accumulated, redistributed, and
	reused over time. A sufficiently large hardware base can
	therefore sustain influence indefinitely, while being
	allocated across many miners, pools, or operational
	identities. The per-block energy cost is a separate
	window-local flow component and is analyzed independently.
	
	Proof-of-Stake replaced computational effort with
	financial capital, eliminating the window-local energy
	flow of PoW while retaining a similar stock structure.
	Financial capital can likewise be accumulated,
	redistributed, and reused over time.
	Delegation mechanisms allow a concentrated capital base to
	support influence across many validator identities without
	requiring proportional renewal costs in each participation
	window~\cite{kiayias2017ouroboros,saleh2021blockchain}.
	The outcome is familiar: as of 2023, Lido Finance alone
	controlled approximately thirty percent of all staked
	ETH~\cite{lido2023concentration}.
	
	The same structural argument extends beyond PoW and PoS.
	Whenever a resource can be accumulated, redistributed, and
	reused, a concentrated resource base can sustain influence
	over an arbitrarily long horizon at cost that grows far
	more slowly than the duration of participation.
	Moreover, the same concentration of resources can often be
	projected through many nominally distinct participants.
	Pooling, delegation, and related coordination mechanisms
	allow a single underlying concentration of control to
	appear as a large and diverse population of participants,
	obscuring the relationship between visible participation
	and effective control.
	
	From this perspective, concentration is not merely a
	byproduct of specific protocol choices. It may emerge from
	the properties of the resource itself. The central question
	is therefore not only how influence is allocated, but
	whether the underlying resource permits influence
	concentration to be sustained at sublinear cost over time.

	\subsection{The Missing Question}
	\label{ssec:missing}
	
	The observation above points to a question that existing
	analyses rarely address:
	
	\begin{quote}
		\emph{What structural properties of a security resource
			determine whether sustained influence can be achieved at
			sublinear cost over time?}
	\end{quote}
	
	Most deployed systems have approached concentration as a
	protocol-design problem. Mining-pool dominance, validator
	concentration, and staking monopolies are typically
	attributed to imperfect incentives, governance rules, or
	participant behavior. This perspective naturally motivates
	new protocol mechanisms intended to improve decentralization.
	
	Our starting point is different. Rather than asking how a
	protocol allocates influence, we ask whether the underlying
	resource itself permits influence to be accumulated,
	redistributed, and reused over time. If so, concentration
	may persist across many different protocol designs because
	it is enabled by the structural properties of the resource.
	
	We formalize this question through the adversarial cost
	function $C(s,T)$: the minimum expenditure required to
	sustain influence proportional to $s$ independent
	participation units over $T$ consecutive windows, whether
	through identity replication, delegation, or pooling.
	Two regimes are particularly important. When
	$C(s,T)=o(sT)$, influence can be sustained through
	resource amortization, allowing a concentrated resource
	base to support long-term influence at sublinear cost.
	When $C(s,T)=\Omega(sT)$, sustained influence requires
	cost that grows proportionally with both scale and time.
	
	The central goal of this paper is to identify which
	resource properties lead to each regime. To do so, we
	develop an axiomatic taxonomy linking resource structure
	to the scaling behavior of $C(s,T)$. The resulting
	framework distinguishes resource classes that admit
	influence amortization from those that enforce linear
	cost, providing a resource-centric perspective on
	decentralization and Sybil resistance.

	\subsection{This Paper}
	\label{ssec:thispaper}
	
	This paper provides an axiomatic resource taxonomy for
	permissionless systems and establishes a sharp asymptotic
	separation between two resource classes.
	
	The \emph{negative half of the taxonomy}
	(Theorem~\ref{thm:impossible}) identifies a class of
	resource-mechanism pairs that admit \emph{influence
		amortization}. Specifically, any resource-mechanism pair
	satisfying divisibility, additivity of influence, temporal
	reusability, and identity transferability yields
	$C(s,T)=o(sT)$.
	Mining hardware capacity, financial capital, storage, and
	delegated stake all fall into this class. (PoW energy,
	being window-local and non-reusable, does not satisfy
	temporal reusability and is analyzed separately as a flow
	component.)
	
	The \emph{positive half of the taxonomy}
	(Theorem~\ref{thm:linear}) introduces
	\emph{throughput-bounded resources}---in which each
	participation unit requires an independent channel, the
	channel output is non-transferable, and allocation does not
	carry over across time windows. We show that such resources
	enforce $C(s,T)=\Omega(sT)$: each additional unit of
	sustained influence incurs marginal cost
	$\Delta(s,T)=\Omega(T)$, growing with the time horizon.
	
	Together, these two halves establish a sharp asymptotic
	separation: parallelizable resources admit influence
	concentration at stock cost $O(s)$, with no
	$T$-proportional renewal cost, whereas throughput-bounded
	resources preclude such amortization.
	
	The separation yields a direct design rule---the
	\emph{Resource Substitution Theorem}: enforcing linear
	cost of influence concentration requires grounding
	participation in a resource that violates at least one
	parallelizability property. Within the scope of
	Definition~\ref{def:parallelizable}, achieving this
	boundary requires altering at least one of the structural
	properties of the underlying resource.

	\subsection{Contributions}
	\label{ssec:contributions}
	
	This paper develops an axiomatic framework for reasoning
	about influence concentration in permissionless systems.
	Rather than focusing on specific consensus protocols, we
	study how the structural properties of security resources
	shape the cost of sustaining influence over time. Our
	contributions are as follows.
	
	\begin{enumerate}
		
		\item \textbf{An adversarial cost framework for influence concentration.}
		
		We introduce the adversarial cost function $C(s,T)$,
		which measures the minimum expenditure required to
		sustain influence equivalent to $s$ independent
		participants over a time horizon of length $T$.
		The framework separates one-time stock costs from
		recurring flow costs and provides a common basis for
		analyzing influence concentration across heterogeneous
		resource classes.
		
		\item \textbf{An axiomatic resource taxonomy for adversarial cost scaling.}
		
		We identify a set of structural properties---divisibility,
		additivity of influence, temporal reusability, identity
		transferability, and throughput boundedness---that govern
		how adversarial cost scales over time. The taxonomy links
		resource structure directly to the asymptotic behavior of
		$C(s,T)$, allowing different resource classes to be
		compared within a common analytical framework.
		
		\item \textbf{A structural upper bound on influence cost under parallelizable resources.}
		
		We prove that any resource-mechanism pair satisfying the
		parallelizability properties admits influence amortization,
		yielding $C(s,T)=o(sT)$. Consequently, a resource base of
		size $O(s)$ can sustain $s$-unit influence over an
		arbitrarily long horizon without costs that grow
		proportionally with time. The result depends only on the
		structural properties of the resource-mechanism pair and
		is independent of the specific consensus protocol built
		upon it.
		
		\item \textbf{A complementary linear-cost resource class and a design principle for decentralization.}
		
		We identify throughput-bounded, non-transferable,
		window-local resources as a class that enforces
		$C(s,T)=\Omega(sT)$ and establish a sharp asymptotic
		separation from parallelizable resources. This separation
		yields the \emph{Resource Substitution Theorem}:
		achieving linear cost of influence concentration requires
		violating at least one of the structural properties
		associated with parallelizability. Within the scope of
		Definition~\ref{def:parallelizable}, this boundary cannot
		be crossed while preserving all four parallelizability
		properties.
		
	\end{enumerate}

	\subsection{Paper Organization}
	\label{ssec:organization}
	
	Section~\ref{sec:related} reviews related work.
	Section~\ref{sec:model} presents the system model.
	Section~\ref{sec:framework} develops the resource framework.
	Section~\ref{sec:impossibility} proves the sublinear cost bound.
	Section~\ref{sec:lowerbound} proves the linear lower bound.
	Section~\ref{sec:separation} establishes the structural
	separation, the Resource Substitution Theorem,
	and hybrid system behavior.
	Section~\ref{sec:instantiations} classifies concrete resource types.
	Section~\ref{sec:evaluation} illustrates the asymptotic separation.
	Section~\ref{sec:limitations} discusses limitations.
	Section~\ref{sec:conclusion} concludes.
	
	\section{Related Work}
	\label{sec:related}
	
	Prior work on Sybil resistance and permissionless systems
	addresses several distinct security questions.
	We survey the main research directions and identify
	the gap this paper fills.
	Table~\ref{tab:positioning} summarizes the distinction.
	
	\subsection{Resource-Based Sybil Defenses}
	\label{ssec:rw_resource}
	
	Douceur~\cite{douceur2002sybil} established that in the
	absence of a trusted authority, identity replication in
	open systems cannot be prevented.
	Nakamoto~\cite{nakamoto2008bitcoin} proposed binding
	participation to a scarce resource---computational
	effort---as a practical response to this impossibility.
	Formal analyses of Nakamoto-style
	consensus~\cite{garay2015backbone,ren2019nakamoto,pass2017hybrid,pass2017fruitchains}
	establish chain-growth, common-prefix, and chain-quality
	guarantees expressed as thresholds on adversarial hash power.
	Proof-of-Stake protocols---including
	Ouroboros~\cite{kiayias2017ouroboros,kiayias2018ouroboros,kiayias2020genesis},
	Snow White~\cite{pass2016snowwhite},
	and Algorand~\cite{gilad2017algorand}---replace
	computational effort with financial capital and derive
	analogous guarantees as thresholds on adversarial stake.
	These protocols collectively instantiate the
	aggregate-threshold adversarial model formalized in
	foundational Byzantine fault-tolerance
	work~\cite{lamport1982byzantine,castro1999pbft};
	Algorand additionally employs verifiable random
	functions~\cite{micali1999vrf} for committee selection.
	
	These analyses consistently reason about
	\emph{aggregate} adversarial resource ownership.
	Security guarantees take the form: if the adversary
	controls less than a threshold fraction of the total
	resource, the protocol is safe.
	They do not characterize whether a stock of size $O(s)$
	can sustain $s$-unit influence over $T$ windows at
	cost $o(sT)$---that is, without $T$-proportional renewal.
	The present paper addresses this question directly.
	
	\subsection{Empirical and Systems-Level Analyses}
	\label{ssec:rw_empirical}
	
	Empirical work documents the centralization patterns
	that resource-based mechanisms produce in practice.
	Gencer et al.~\cite{gencer2018decentralization} measure
	mining pool concentration in Bitcoin and Ethereum.
	Rosenfeld~\cite{rosenfeld2011analysis,rosenfeld2014analysis}
	analyzes pooled mining reward systems and hash-rate-based
	double-spending dynamics.
	Bonneau et al.~\cite{bonneau2015sok} and
	Narayanan et al.~\cite{narayanan2016bitcoinbook} survey
	the broader security and deployment landscape of Bitcoin.
	Gervais et al.~\cite{gervais2016security,gervais2014isbitcoin}
	quantify security-performance tradeoffs and decentralization
	properties of deployed proof-of-work systems.
	Work on selfish mining~\cite{eyal2014majority} and
	double-spending~\cite{kroll2013economics} analyzes
	adversarial strategies under aggregate resource models.
	Network-layer studies~\cite{neudecker2019network} document
	how propagation dynamics interact with pool-driven
	resource concentration.
	
	These works describe the \emph{symptoms} of structural
	resource properties---concentration, pooling, delegation---
	but do not formally characterize why such patterns emerge
	as a structural consequence of the resource properties
	considered in this model.
	The present paper provides this characterization.
	
	\subsection{Economic Analyses of Consensus}
	\label{ssec:rw_economic}
	
	Budish~\cite{budish2018economic} and
	Budish et al.~\cite{budish2024economic} analyze the
	economic limits of permissionless consensus, framing
	security in terms of the cost of mounting a majority attack.
	Carlsten et al.~\cite{carlsten2016instability} study
	incentive stability in the absence of block rewards.
	Saleh~\cite{saleh2021blockchain} provides an equilibrium
	analysis of Proof-of-Stake.
	
	These analyses model adversarial capability in terms of
	aggregate resource ownership and study equilibrium
	incentives or corruption costs.
	They do not study whether influence can be concentrated
	at sublinear cost through delegation or pooling---the
	question central to this paper.
	
	\subsection{Graph-Based and Social-Structure Defenses}
	\label{ssec:rw_graph}
	
	SybilGuard~\cite{yu2006sybilguard},
	SybilLimit~\cite{yu2008sybillimit}, and
	SybilControl~\cite{li2017sybilcontrol} bound adversarial
	influence by exploiting the observation that in social
	networks, the number of edges between honest and Sybil
	regions is small relative to internal connectivity.
	SybilInfer~\cite{danezis2009sybilinfer} applies
	probabilistic inference over trust graphs to identify
	adversarial clusters.
	Game-theoretic analyses~\cite{lesniewski2020analysis}
	study equilibrium conditions under which honest participants
	can deter adversarial identity creation.
	
	These approaches address a different question: where
	Sybil identities can be placed within a social graph,
	rather than how the cost of adversarial influence scales
	over time.
	They provide no characterization of $C(s,T)$ as $s$
	or $T$ grows, and their guarantees depend on
	graph-topology assumptions that may not hold in
	practice~\cite{wang2019attack}.
	
	\subsection{Identity-Binding Mechanisms}
	\label{ssec:rw_identity}
	
	Proof-of-Personhood~\cite{borge2017pop} and pseudonym
	parties~\cite{ford2020pseudonym} bind credentials to
	real-world participants to enforce one-person--one-identity
	mappings.
	Reputation systems~\cite{resnick2000reputation} offer a
	related but distinct approach: rather than binding
	participation to verified identity, they tie influence
	to accumulated behavioral history.
	These mechanisms address \emph{admission}---restricting
	the creation of new identities---rather than the sustained
	cost of concentrating influence over time.
	Once admitted, influence may still be governed by a
	transferable, reusable resource.
	
	More recent deployed systems instantiate these admission
	mechanisms directly.
	World ID~\cite{worldcoin2023whitepaper} binds a single
	credential to a unique biological human via biometric
	verification, and maps onto the \emph{human real-time
		participation} class of Table~\ref{tab:taxonomy}: each
	verified human constitutes one per-actor channel.
	Privacy Pass~\cite{davidson2018privacypass} and its
	rate-limited extensions~\cite{chu2023ratelimited} bind
	anonymous credentials to per-user quotas, instantiating
	the \emph{rate-limited participation} class---subject to
	the channel-binding condition of
	Definition~\ref{def:throughput}.
	As with the device-bound and human-participation
	instantiations of Section~\ref{sec:instantiations}, the
	structural guarantees apply only to the extent that the
	underlying credential cannot be resold, shared, or
	generated by automated agents.
	
	\subsection{Capability and Rate-Limiting Architectures}
	\label{ssec:rw_capability}
	
	Hardware-rooted execution environments such as
	Intel SGX~\cite{costan2016intel} enforce per-device
	participation limits that bind resource generation
	to specific hardware.
	Capability-based systems~\cite{levy1984capability}
	and congestion-control mechanisms~\cite{jacobson1988congestion}
	impose per-actor throughput constraints in different
	deployment contexts.
	Bitcoin-NG~\cite{eyal2018bitcoinng} decouples block
	leadership from transaction throughput, instantiating
	a per-leader rate constraint as a protocol-level primitive.
	HoneyBadger BFT~\cite{miller2016honeybadger} makes
	explicit throughput bounds central to its communication
	complexity analysis.
	Alwen et al.~\cite{alwen2019blockchain} study
	sustained space and time complexity of sequential
	proof-of-work constructions.
	
	These works impose per-actor throughput constraints
	but analyze them in the context of access control,
	network stability, or computational hardness---not
	the cost of influence concentration.
	They do not establish the structural separation between
	parallelizable and throughput-bounded resource classes
	that this paper proves.
	
	\subsection{Positioning}
	\label{ssec:rw_positioning}
	
	The common thread across prior work is a focus on
	\emph{aggregate} adversarial capability: how much
	resource does an adversary control, and is it below
	a safety threshold?
	To our knowledge, prior work does not ask whether a
	fixed aggregate resource can be used to concentrate
	influence at sublinear marginal cost---whether through
	identity replication, delegation, or pooling---nor
	establish the structural conditions under which this
	is or is not possible.
	
	This paper provides that characterization.
	We identify the structural properties that determine
	whether adversarial cost scales as $o(sT)$ or
	$\Omega(sT)$, establish that parallelizable resources admit
	sublinear influence cost, introduce the throughput-bounded
	resource class and prove its linear cost guarantee,
	and derive a design principle for systems requiring
	linear cost of influence concentration.
	
	\begin{table*}[htbp]
		\caption{Positioning of This Work Relative to Prior Literature}
		\label{tab:positioning}
		\renewcommand{\arraystretch}{1.4}
		\begin{tabularx}{\linewidth}{@{}p{2.8cm}Xp{3.8cm}c@{}}
			\toprule
			\textbf{Research Line} & \textbf{Representative Works} & \textbf{Primary Security Question} & \textbf{Influence Conc. Cost?} \\
			\midrule
			Sybil attack literature
			& \cite{douceur2002sybil,yu2006sybilguard,yu2008sybillimit,danezis2009sybilinfer}
			& Can adversaries create multiple identities, and how can identity creation be restricted or detected?
			& No \\[4pt]
			Blockchain consensus security
			& \cite{lamport1982byzantine,castro1999pbft,garay2015backbone,ren2019nakamoto,kiayias2017ouroboros,gilad2017algorand}
			& Under what aggregate adversarial resource ownership do safety and liveness properties fail?
			& No \\[4pt]
			Economic analyses of consensus
			& \cite{budish2018economic,budish2024economic,kroll2013economics}
			& What equilibrium incentives and corruption costs arise in resource-weighted consensus?
			& No \\[4pt]
			Graph-based defenses
			& \cite{yu2006sybilguard,yu2008sybillimit,danezis2009sybilinfer}
			& Can Sybil proliferation be constrained via social-network structure?
			& No \\[4pt]
			Identity-binding mechanisms
			& \cite{borge2017pop,ford2020pseudonym}
			& How can systems enforce one-person--one-identity mappings?
			& No \\[4pt]
			Capability / rate-limiting
			& \cite{costan2016intel,levy1984capability,jacobson1988congestion}
			& Can per-actor throughput limits be enforced at the system level?
			& No \\
			\midrule
			\textbf{This paper}
			& ---
			& How do structural properties of security resources determine whether influence can be concentrated at sublinear cost?
			& \textbf{Yes} \\
			\bottomrule
		\end{tabularx}
	\end{table*}
	
	
	\section{System Model}
	\label{sec:model}
	
	We model an open and permissionless system in which
	identity creation is inexpensive and no admission
	authority exists.
	Influence is regulated by a security-weighting resource
	$\mathcal{R}$.
	Our objective is to characterize how structural
	properties of $\mathcal{R}$ determine whether influence
	can be concentrated at sublinear cost---whether through
	identity replication, delegation, or pooling.
	
	Table~\ref{tab:notation} summarizes the core notation.

	\subsection{Identities and Participation}
	\label{ssec:identity}
	
	An identity is a cryptographic key pair that authorizes
	participation.
	A single entity may control arbitrarily many identities
	at negligible cost~\cite{douceur2002sybil}, so influence
	may be acquired through \emph{identity replication} or
	\emph{resource aggregation}---the model captures both
	uniformly.
	
	Let $S_t \subseteq \mathcal{V}_t$ denote the adversarial
	identities at window $t$, with $|S_t| = s_t$.
	The adversary $\mathcal{A}$ is \emph{identity-unbounded}:
	it may create, retire, and coordinate an arbitrary
	number of identities at negligible syntactic cost.
	
	\subsection{Time Model}
	\label{ssec:time}
	
	Time proceeds in discrete windows $t = 1, 2, \ldots$,
	representing the minimal interval over which resource
	allocation and influence are refreshed.
	Security is evaluated over a horizon of $T$ windows.
	
	\subsection{Resource and Influence}
	\label{ssec:resource}
	
	For each identity $v \in \mathcal{V}_t$ and window $t$,
	let $r_v(t) \in \mathbb{R}_{\geq 0}$ denote the
	resource allocated to $v$.
	Influence is determined by a mapping
	\begin{equation}
		W_v(t) = f\bigl(r_v(t)\bigr),
	\end{equation}
	where $f : \mathbb{R}_{\geq 0} \to \mathbb{R}_{\geq 0}$
	is the protocol's resource-to-influence function;
	structural assumptions on $f$ are deferred to
	Section~\ref{sec:framework}.
	The total adversarial influence at window $t$ is
	$W_{\mathcal{A}}(t) = \sum_{v \in S_t} W_v(t)$.
	
	\begin{definition}[Activation Threshold]
		\label{def:threshold}
		There exists a constant $r_{\min} > 0$ such that
		an identity is active in window $t$ only if
		$r_v(t) \geq r_{\min}$.
	\end{definition}
	
	\subsection{Adversarial Model}
	\label{ssec:adversary}
	
	The adversary $\mathcal{A}$ is computationally
	unbounded but constrained solely by the structural
	properties of $\mathcal{R}$.
	It may instantiate, retire, and coordinate identities;
	allocate and reallocate resource; reuse resource across
	windows; and aggregate influence through delegation or
	pooling---whenever the resource structure permits.
	
	\subsection{Adversarial Cost}
	\label{ssec:cost}
	
	We distinguish between \emph{resource allocation} and
	\emph{new expenditure}, since the two differ for
	temporally reusable resources.
	
	\begin{definition}[Allocation and Expenditure]
		\label{def:allocation_expenditure}
		For each identity $v \in S_t$ and window $t$:
		\begin{itemize}
			\item $r_v(t) \geq 0$ is the resource
			\emph{allocated} to $v$---the amount drawn upon,
			whether freshly acquired or carried forward.
			\item $e_v(t) \geq 0$ is the resource
			\emph{newly expended}---the incremental cost
			actually incurred in that window.
		\end{itemize}
		For a reusable stock resource, $e_v(t) = 0$ for
		$t \geq 2$; for a window-local resource,
		$e_v(t) \geq r_{\min}$ must hold every window.
		Let $E_{\mathcal{A}}(t) = \sum_{v \in S_t} e_v(t)$.
	\end{definition}
	
	\begin{definition}[Adversarial Cost Function]
		\label{def:cost}
		Let $W_{\mathrm{unit}} = f(r_{\min})$.
		$C(s, T)$ denotes the minimum total \emph{new
			expenditure} required to achieve influence
		equivalent to that of $s$ independent participants
		over $T$ consecutive windows:
		\begin{equation}
			\label{eq:cost}
			C(s, T) \;=\; \inf \left\{
			A\!\left(\{r_v(1)\}\right)
			+ \sum_{t=1}^{T} E_{\mathcal{A}}(t)
			+ h(s,T)
			\;\Bigg|\;
			W_{\mathcal{A}}(t) \geq s \cdot W_{\mathrm{unit}},\;
			\forall\, t \in \{1,\ldots,T\}
			\right\},
		\end{equation}
		where $A(\{r_v(1)\})$ is the one-time acquisition
		cost of the initial stock and $h(s,T)$ is the
		coordination overhead.
	\end{definition}
	
	Because $C(s,T)$ is defined as an infimum over all
	admissible adversarial strategies, the model already
	captures adversaries that adapt identity allocation,
	resource distribution, and participation patterns over
	the horizon.
	The limitation discussed in Section~\ref{ssec:lim_scope}
	concerns richer dynamic models in which the
	resource-mechanism pair itself evolves in response to
	adversarial behavior.
	
	The stock--flow distinction is central.
	For a reusable resource, $\sum E_{\mathcal{A}}(t) = 0$
	once stock is acquired, so $C(s,T)$ grows only with
	$s$, not $T$.
	For a window-local resource,
	$\sum E_{\mathcal{A}}(t) \geq s \cdot T \cdot r_{\min}$,
	so $C(s,T)$ grows with both.
	
	We decompose $C(s,T)$ as
	\begin{equation}
		C(s, T) = C_{\mathrm{stock}}(s)
		+ C_{\mathrm{flow}}(s, T)
		+ h(s, T),
	\end{equation}
	where $C_{\mathrm{stock}}(s) = A(\{r_v(1)\})$
	is the one-time acquisition cost,
	$C_{\mathrm{flow}}(s,T) = \sum_{t=1}^{T} E_{\mathcal{A}}(t)$
	is the cumulative renewal expenditure, and $h(s,T)$
	is coordination overhead.
	
	\begin{assumption}[Amortized Coordination]
		\label{asm:coordination}
		The coordination overhead satisfies:
		\begin{enumerate}
			\item[(a)] $h(s, T) = o(sT)$; and
			\item[(b)] $h(s, T) - h(s-1, T) = o(T)$
			for each fixed $s$ as $T \to \infty$.
		\end{enumerate}
	\end{assumption}
	
	Both conditions are consistent with coordination
	infrastructure observed in practice, including mining
	pools, validator delegation services, and cloud-managed
	orchestration~\cite{rosenfeld2011analysis,
		bonneau2015sok,saleh2021blockchain}.
	Condition~(a) is used in Theorem~\ref{thm:impossible};
	condition~(b) bounds the marginal cost $\Delta(s,T)$.
	
	\subsection{Scaling Regimes}
	\label{ssec:regimes}
	
	\begin{definition}[Scaling Regimes]
		\label{def:regimes}
		A mechanism exhibits \emph{sublinear scaling} if
		$C(s, T) = o(sT)$, and \emph{linear scaling} if
		$C(s, T) = \Omega(sT)$.
	\end{definition}
	
	Under sublinear scaling, a resource base of size
	$O(s)$ can sustain influence equivalent to $s$
	participants over an arbitrarily long horizon.
	
	Under linear scaling, sustaining influence equivalent
	to $s$ participants over $T$ windows requires cost
	growing as $\Omega(sT)$.
	
	The central question of this paper is which scaling
	regime a given resource $\mathcal{R}$ induces.
	Section~\ref{sec:framework} provides the structural
	answer.
	
	\begin{table}[H]
		\caption{Core Notation}
		\label{tab:notation}
		\renewcommand{\arraystretch}{1.25}
		\begin{tabular}{@{}ll@{}}
			\toprule
			\textbf{Symbol} & \textbf{Meaning} \\
			\midrule
			$v$                  & Identity (key pair / participant instance) \\
			$t$                  & Time window index \\
			$T$                  & Time horizon (number of windows) \\
			$s$                  & Influence level in participant equivalents \\
			$r_v(t)$             & Resource allocated to identity $v$ at window $t$ \\
			$e_v(t)$             & New expenditure by identity $v$ in window $t$ \\
			$R_{\mathcal{A}}(t)$ & Total adversarial resource allocated at window $t$ \\
			$E_{\mathcal{A}}(t)$ & Total new adversarial expenditure in window $t$ \\
			$W_v(t)$             & Influence of identity $v$ at window $t$ \\
			$W_{\mathrm{unit}}$  & Baseline influence: $f(r_{\min})$ \\
			$f(\cdot)$           & Resource-to-influence mapping \\
			$S_t$                & Set of adversarial identities at window $t$ \\
			$C(s,T)$             & Min.\ cost to achieve influence equivalent to $s$ participant equivalents over $T$ windows \\
			$\Delta(s,T)$        & Marginal cost: $C(s,T) - C(s-1,T)$ \\
			$r_{\min}$           & Activation threshold \\
			$\tau$               & Per-channel throughput bound \\
			$h(s,T)$             & Coordination overhead \\
			\bottomrule
		\end{tabular}
	\end{table}
	\section{Resource Framework}
	\label{sec:framework}
	
	We characterize security-weighting resources by their
	structural properties---the constraints they impose on
	division, aggregation, reuse, and transfer across
	identities and time.
	These properties play a central role in determining
	whether influence concentration can be sustained at
	sublinear cost.
	
	\subsection{Structural Properties}
	\label{ssec:properties}
	
	\begin{definition}[Divisibility]
		\label{def:divisibility}
		$\mathcal{R}$ is \emph{divisible} if for any allocation
		$r \geq 0$ and integer $k \geq 1$, there exist
		$r_1, \ldots, r_k \geq 0$ with $\sum_{i=1}^{k} r_i = r$,
		and the adversary may distribute $r_i$ to identity $i$
		without additional structural cost.
	\end{definition}
	
	\begin{definition}[Additivity of Influence]
		\label{def:additivity}
		A mechanism \emph{induces additive influence} if
		\begin{equation}
			\sum_{v \in S_t} f\bigl(r_v(t)\bigr)
			\;=\;
			f\!\left(\sum_{v \in S_t} r_v(t)\right).
		\end{equation}
		Total influence depends only on aggregate allocation,
		not on how it is partitioned across identities.
	\end{definition}
	
	\begin{remark}
		Definition~\ref{def:additivity} is a property of
		the influence function $f$, not of the resource alone.
		Mechanisms that deliberately penalize concentration---
		such as quadratic voting, where $f(r) = \sqrt{r}$---
		do not satisfy it and fall outside the scope of
		Theorem~\ref{thm:impossible}.
		Theorem~\ref{thm:impossible} applies only to
		resource-mechanism pairs for which influence is
		additive in this sense; non-additive
		resource-to-influence mappings may induce different
		cost-scaling behavior.
		Characterizing the scaling regimes admitted by
		non-additive mechanisms is outside the scope of the
		present work.
	\end{remark}
	
	Additivity is not posited as a universal property of
	resource-to-influence mappings, but as the property
	that characterizes the influence mechanisms underlying
	the major resource-weighted systems motivating this
	work---including Proof-of-Work mining power,
	stake-weighted participation, delegated stake, and
	pooled resource ownership.
	Theorem~\ref{thm:impossible} therefore characterizes
	this practically dominant class of resource-mechanism
	pairs; mechanisms with non-additive influence functions
	are addressed by the scope limitation above.
	
	\begin{definition}[Temporal Reusability]
		\label{def:reusability}
		$\mathcal{R}$ is \emph{temporally reusable} if an
		allocation obtained in window $t$ may be applied in
		window $t+1$ without renewed identity-bound expenditure.
	\end{definition}
	
	\begin{definition}[Identity Transferability]
		\label{def:transferability}
		$\mathcal{R}$ is \emph{identity-transferable} if an
		allocation $r_v(t)$ assigned to identity $v$ may be
		reassigned to any other identity $v'$ without
		proportional cost.
	\end{definition}
	
	\begin{definition}[Throughput-Bounded Resource]
		\label{def:throughput}
		$\mathcal{R}$ is \emph{throughput-bounded} if each
		identity $v$ is associated with a dedicated
		participation channel $C_v$ satisfying:
		\begin{enumerate}
			\item \emph{Per-channel rate limit.}
			There exists $\tau > 0$ such that
			$r_v(t) \leq \tau$ for all $v$ and $t$.
			\item \emph{Non-transferability.}
			$r_v(t)$ cannot be simultaneously credited
			to any other identity $v' \neq v$.
			\item \emph{Window-locality.}
			Allocation in window $t$ does not carry over
			to window $t+1$; sustained participation
			requires renewed expenditure in every window.
		\end{enumerate}
		We assume $0 < r_{\min} \leq \tau$ throughout.
	\end{definition}
	
	\subsection{Resource Classes}
	\label{ssec:classes}
	
	\begin{definition}[Parallelizable Resource]
		\label{def:parallelizable}
		$\mathcal{R}$ is \emph{parallelizable} if it satisfies
		divisibility, temporal reusability, and identity
		transferability, and the induced mechanism satisfies
		additivity of influence.
	\end{definition}
	
	\begin{definition}[Throughput-Bounded Non-Parallelizable Resource]
		\label{def:nonparallelizable}
		$\mathcal{R}$ is a \emph{throughput-bounded
			non-parallelizable resource} if it satisfies
		Definition~\ref{def:throughput}.
	\end{definition}
	
	\subsection{Resource Taxonomy}
	\label{ssec:taxonomy}
	
	Table~\ref{tab:taxonomy} classifies representative
	resource types by the structural properties they satisfy.
	Proof-of-Work is decomposed into its capacity component
	(mining hardware: parallelizable stock) and its
	operational component (energy: window-local flow).
	
	\begin{table*}[htbp]
		\caption{Structural Taxonomy of Security Resources.
			PoW is decomposed into its capacity component
			(hardware: parallelizable) and operational component
			(energy: non-reusable flow).
			\textbf{Add.}\ indicates whether the mechanism
			induces additive influence
			(Definition~\ref{def:additivity}).}
		\label{tab:taxonomy}
		\renewcommand{\arraystretch}{1.35}
		\begin{tabularx}{\linewidth}{@{}p{2.4cm}ccccp{2.8cm}c@{}}
			\toprule
			\textbf{Resource} & \textbf{Div.} & \textbf{Add.}
			& \textbf{Reuse} & \textbf{Transfer}
			& \textbf{Notes} & \textbf{Scaling} \\
			\midrule
			PoW: hardware (capacity)
			& \checkmark & \checkmark & \checkmark & \checkmark
			& Parallelizable stock; mining pools instantiate
			amortization
			& $o(sT)$ \\[3pt]
			PoW: energy (operational)
			& \checkmark & \checkmark & $\times$ & \checkmark
			& Window-local flow; renewed each block interval
			& Linear in $T$\textsuperscript{†} \\[3pt]
			Financial stake (PoS)
			& \checkmark & \checkmark & \checkmark & \checkmark
			& Parallelizable stock; delegation markets
			instantiate amortization
			& $o(sT)$ \\[3pt]
			Social-graph trust
			& Partial & Partial & Partial & Partial
			& Guarantees depend on graph topology;
			no explicit $C(s,T)$ bound
			& Model-dep. \\[3pt]
			Device-bound execution
			& $\times$ & $\times$ & $\times$ & $\times$
			& Each device constitutes one independent channel
			per window
			& $\Omega(sT)$ \\[3pt]
			Human real-time participation
			& $\times$ & $\times$ & $\times$ & $\times$
			& Cognitive and temporal throughput limits bound
			channel capacity
			& $\Omega(sT)$ \\[3pt]
			Rate-limited participation\textsuperscript{$\ddagger$}
			& $\times$ & $\times$ & $\times$ & $\times$
			& Qualifies only when bound to non-transferable
			execution channels
			& $\Omega(sT)$ \\
			\bottomrule
			\multicolumn{7}{@{}l}{\footnotesize
				\textsuperscript{†}Energy is non-reusable but
				transferable; it falls outside the parallelizable
				class of Theorem~\ref{thm:impossible}.} \\
			\multicolumn{7}{@{}l}{\footnotesize
				\textsuperscript{$\ddagger$}Account-level rate limits
				without channel binding are bypassable through
				identity replication.} \\
		\end{tabularx}
	\end{table*}
	
	\section{Sublinear Cost Under Parallelizable Resources}
	\label{sec:impossibility}
	
	We prove that no resource-mechanism pair satisfying
	the four parallelizability properties can enforce
	linear cost of influence concentration.
	
	\begin{theorem}[Amortization Under Parallelizable Resources]
		\label{thm:impossible}
		Let $(\mathcal{R}, f)$ be a parallelizable
		resource-mechanism pair
		(Definition~\ref{def:parallelizable}).
		Under Assumption~\ref{asm:coordination},
		\begin{equation}
			C(s, T) = o(sT).
		\end{equation}
		In particular, $\Delta(s, T) = C(s,T) - C(s-1,T)
		= o(T)$: the marginal cost of influence concentration
		does not grow with the time horizon.
	\end{theorem}
	
	\begin{proof}
		By divisibility and additivity of influence, any
		aggregate allocation $r$ may be partitioned across
		$s$ identities without changing total influence:
		$\sum_v f(r_v) = f(r)$ for any partition
		$\{r_v\}$ with $\sum_v r_v = r$.
		By temporal reusability, an allocation
		$r_v(1) = r_{\min}$ acquired in window~$1$ remains
		available in all subsequent windows without renewed
		expenditure.
		Therefore the adversary achieves $s$-unit influence
		at total cost
		\begin{equation}
			\label{eq:stock_bound}
			C(s, T) \;\leq\; s \cdot r_{\min} + h(s, T),
		\end{equation}
		where $C_{\mathrm{flow}}(s,T) = 0$ since no renewal
		expenditure is required after window~$1$.
		By Assumption~\ref{asm:coordination}(a),
		$h(s,T) = o(sT)$, so
		\begin{equation}
			\frac{C(s,T)}{sT}
			\leq \frac{r_{\min}}{T} + \frac{h(s,T)}{sT}
			\to 0
			\quad \text{as } s, T \to \infty.
		\end{equation}
		From~\eqref{eq:stock_bound},
		\begin{equation}
			\Delta(s,T) \;\leq\; r_{\min} + h(s,T) - h(s-1,T).
		\end{equation}
		By Assumption~\ref{asm:coordination}(b),
		$h(s,T) - h(s-1,T) = o(T)$ for each fixed $s$,
		so $\Delta(s,T) = o(T)$.
	\end{proof}
	
	The result is structural: for any parallelizable
	resource-mechanism pair, the $T$-scaling flow cost
	$C_{\mathrm{flow}}(s,T)$ vanishes once the initial
	stock is acquired, so sustaining $s$-unit influence
	requires only $O(s)$ cost regardless of the time
	horizon.
	Enforcing linear cost therefore requires violating
	at least one of the parallelizability properties,
	whether through the choice of resource or through
	protocol mechanisms that prevent those properties
	from being realized in practice.
	\section{Linear Lower Bound Under Throughput-Bounded Resources}
	\label{sec:lowerbound}
	
	We prove the complementary result: throughput-bounded
	non-parallelizable resources structurally enforce
	$C(s,T) = \Omega(sT)$.
	
	\begin{theorem}[Linear Lower Bound]
		\label{thm:linear}
		Let $\mathcal{R}$ be a throughput-bounded
		non-parallelizable resource
		(Definition~\ref{def:nonparallelizable}).
		Then there exists $c > 0$ such that
		\begin{equation}
			C(s, T) \;\geq\; c \cdot sT,
		\end{equation}
		and in particular $C(s,T) = \Omega(sT)$.
	\end{theorem}
	
	\begin{proof}
		Set $c = r_{\min}$.
		By non-transferability (condition~2 of
		Definition~\ref{def:throughput}), each active identity
		requires its own allocation: $r_v(t) \geq r_{\min}$
		must hold independently for each of the $s$ identities.
		By window-locality (condition~3), allocation does not
		carry forward, so new expenditure equals allocation
		in every window.
		Therefore $E_{\mathcal{A}}(t) \geq s \cdot r_{\min}$
		for every $t$, and
		\begin{equation}
			C(s, T) \;\geq\; \sum_{t=1}^{T} E_{\mathcal{A}}(t)
			\;\geq\; s \cdot T \cdot r_{\min}
			\;=\; c \cdot sT.
		\end{equation}
	\end{proof}
	
	Unlike parallelizable resources, expenditure cannot
	be amortized across identities or across time.
	Each active identity requires renewed expenditure
	in every window, causing both the number of
	participants and the time horizon to contribute
	linearly to cost.
	
	\section{Structural Separation and Design Consequences}
	\label{sec:separation}
	
	Theorems~\ref{thm:impossible} and~\ref{thm:linear}
	characterize two fundamentally different classes of
	security-weighting resources.
	We now combine these results into a single structural
	separation and derive its implications for mechanism
	design.
	
	\subsection{Structural Separation}
	\label{ssec:separation_theorem}
	
	\begin{theorem}[Structural Separation]
		\label{thm:separation}
		Let $(\mathcal{R}, f)$ be a security-weighting
		resource-mechanism pair.
		\begin{enumerate}
			\item If $\mathcal{R}$ is parallelizable, then
			$C(s,T) = o(sT)$.
			\item If $\mathcal{R}$ is a throughput-bounded
			non-parallelizable resource, then
			$C(s,T) = \Omega(sT)$.
		\end{enumerate}
		The two classes induce asymptotically disjoint
		scaling regimes.
	\end{theorem}
	
	\begin{proof}
		Part~(1) is Theorem~\ref{thm:impossible}.
		Part~(2) is Theorem~\ref{thm:linear}.
		Since $o(sT)$ and $\Omega(sT)$ are asymptotically
		disjoint, no resource can exhibit both regimes.
	\end{proof}
	
	Theorem~\ref{thm:separation} establishes that the
	distinction between amortizable and non-amortizable
	influence is not quantitative but structural.
	Parallelizable resources admit sublinear scaling,
	whereas throughput-bounded non-parallelizable resources
	enforce linear scaling.
	No resource satisfying one class definition can
	simultaneously exhibit the asymptotic behavior of
	the other.
	
	\subsection{Resource Substitution Theorem}
	\label{ssec:rsp}
	
	\begin{theorem}[Resource Substitution Theorem]
		\label{thm:substitution}
		Any mechanism targeting $C(s,T) = \Omega(sT)$ must
		ground participation in a resource that violates at
		least one of the following properties:
		divisibility, additivity of influence,
		temporal reusability, identity transferability.
	\end{theorem}
	
	\begin{proof}
		If the underlying resource satisfies all four
		properties it is parallelizable
		(Definition~\ref{def:parallelizable}), and
		Theorem~\ref{thm:impossible} gives
		$C(s,T) = o(sT)$, a contradiction.
	\end{proof}
	
	\begin{corollary}
		\label{cor:protocol}
		For any resource-mechanism pair satisfying
		Definition~\ref{def:parallelizable} and
		Assumption~\ref{asm:coordination},
		Theorem~\ref{thm:impossible} implies that no
		protocol-level modification preserving these
		properties can enforce $C(s,T) = \Omega(sT)$.
		Enforcing linear cost therefore requires violating
		at least one parallelizability property.
	\end{corollary}
	
	The Resource Substitution Theorem is the primary
	design consequence of this work.
	Within the class of parallelizable resources, linear
	influence cost cannot be achieved through protocol
	engineering alone.
	Any mechanism seeking $C(s,T) = \Omega(sT)$ must
	ultimately rely on a resource that violates at least
	one of the four parallelizability properties.
	The boundary between amortizable and non-amortizable
	influence is therefore determined by resource
	structure rather than by protocol rules.
	
	\section{Instantiations and Classification}
	\label{sec:instantiations}
	
	This section instantiates the framework by classifying
	representative resource types according to the
	structural properties of
	Definitions~\ref{def:divisibility}--\ref{def:throughput}.
	
	\subsection{Per-Actor Throughput Channels}
	\label{ssec:channels}
	
	\begin{definition}[Per-Actor Throughput Channel]
		\label{def:channel}
		A participation channel $C_u$ associated with
		actor $u$ is a \emph{per-actor throughput channel}
		if there exists $\tau > 0$ such that for every
		window $t$, $\mathrm{rate}_t(C_u) \leq \tau$, and
		$r_v(t) \leq \mathrm{rate}_t(C_u)$ for any identity
		$v$ backed by $C_u$.
	\end{definition}
	
	\begin{proposition}[Per-Actor Channel Classification]
		\label{prop:channel}
		Any resource derived from per-actor throughput
		channels is a throughput-bounded non-parallelizable
		resource (Definition~\ref{def:nonparallelizable}).
		Consequently, $C(s,T) = \Omega(sT)$ by
		Theorem~\ref{thm:linear}.
	\end{proposition}
	
	\begin{proof}
		Throughput boundedness follows from the rate limit
		$\tau$ of Definition~\ref{def:channel}.
		Non-transferability holds because $C_u$'s output
		cannot be credited to an identity under a different
		actor.
		Window-locality holds because $\mathrm{rate}_t(C_u)$
		does not carry over to window $t+1$.
		Therefore Definition~\ref{def:throughput} is
		satisfied, and the resource belongs to the
		throughput-bounded non-parallelizable class.
	\end{proof}
	
	\subsection{Classification of Deployed Resources}
	\label{ssec:walkthrough}
	
	\noindent\textbf{Proof-of-Work (hardware component).}\ 
	Mining hardware is divisible, transferable,
	temporally reusable, and induces additive influence.
	It therefore satisfies the parallelizable resource
	class, implying $C(s,T) = o(sT)$ by
	Theorem~\ref{thm:impossible}.
	Mining pools provide an empirical manifestation of
	this amortization~\cite{gencer2018decentralization,
		rosenfeld2011analysis}.
	
	\noindent\textbf{Proof-of-Stake (financial capital).}\ 
	Financial stake satisfies the same four
	parallelizability properties.
	It therefore induces $C(s,T) = o(sT)$ by
	Theorem~\ref{thm:impossible}.
	Large staking pools and delegation systems provide
	empirical manifestations of this amortization pattern.
	
	\noindent\textbf{Device-bound execution.}\ 
	Each trusted execution environment constitutes one
	per-actor throughput channel.
	By Proposition~\ref{prop:channel},
	$C(s,T) = \Omega(sT)$, subject to the deployment
	fidelity conditions discussed in
	Section~\ref{ssec:lim_enforcement}.
	The classification is structural: it characterizes
	the resource as defined by
	Definition~\ref{def:throughput}, not any particular
	hardware implementation.
	In practice, SGX-like environments may deviate from
	this idealization through relay attacks, remote
	attestation by a third party, virtualization of the
	enclave, or leasing and outsourcing of attestation
	capacity to other actors---each of which can weaken
	non-transferability or window-locality.
	To the extent that such deviations are absent, the
	resource approximates a throughput-bounded channel
	under the framework's assumptions; the consequences
	of partial deviation are discussed further in
	Section~\ref{sec:limitations}.
	
	\noindent\textbf{Human real-time participation.}\ 
	Each participant constitutes one channel with
	intrinsic per-window throughput limits.
	By Proposition~\ref{prop:channel},
	$C(s,T) = \Omega(sT)$, subject to the same
	deployment fidelity conditions.
	As with device-bound execution, this classification
	is structural and reflects the idealized properties
	of Definition~\ref{def:throughput} rather than any
	specific deployment.
	Deployed human-participation systems may depart from
	this idealization through credential markets that
	resell verified identities, outsourcing of
	participation tasks to third parties, or automation
	that mimics human responses---each of which can
	erode non-transferability or the per-channel rate
	limit.
	Provided such practices remain limited, the resource
	approximates a throughput-bounded channel under the
	framework's assumptions; Section~\ref{sec:limitations}
	discusses the open problem of characterizing this
	approximation more precisely.
	
	\noindent\textbf{Rate-limited network participation.}
	Account-level rate limits qualify as throughput-bounded
	only when bound to non-transferable execution channels.
	Otherwise, limits attached solely to syntactic identities
	can be circumvented through identity replication.
	\section{ILLUSTRATING THE STRUCTURAL SEPARATION}
	\label{sec:evaluation}
	
	The following figures illustrate the asymptotic
	behavior implied by Theorems~\ref{thm:impossible}
	and~\ref{thm:linear} for a representative coordination
	overhead $h(s,T)=s+T$ satisfying
	Assumption~\ref{asm:coordination}; they are not
	empirical measurements.
	With this choice,
	
	\[
	\frac{C_{\mathrm{par}}(s,T)}{sT}
	=
	\frac{r_{\min}+1}{T}
	+
	\frac{1}{s},
	\]
	
	which vanishes only in the joint limit $s,T\to\infty$, not along fixed-$s$ or fixed-$T$ slices.
	
	\subsection{Cost Scaling and Joint-Asymptotic Convergence}
	\label{ssec:eval_scaling}
	
	Figure~\ref{fig:joint_asymptotic} shows the normalised ratio along the diagonal $s=T$. Along this diagonal,
	
	\[
	\frac{C_{\mathrm{par}}(s,s)}{s^2}
	=
	\frac{r_{\min}+2}{s}
	\to 0
	\qquad\text{as } s\to\infty,
	\]
	
	confirming Theorem~\ref{thm:impossible} in the joint regime. The throughput-bounded ratio remains fixed at $r_{\min}$ (Theorem~\ref{thm:linear}).
	
	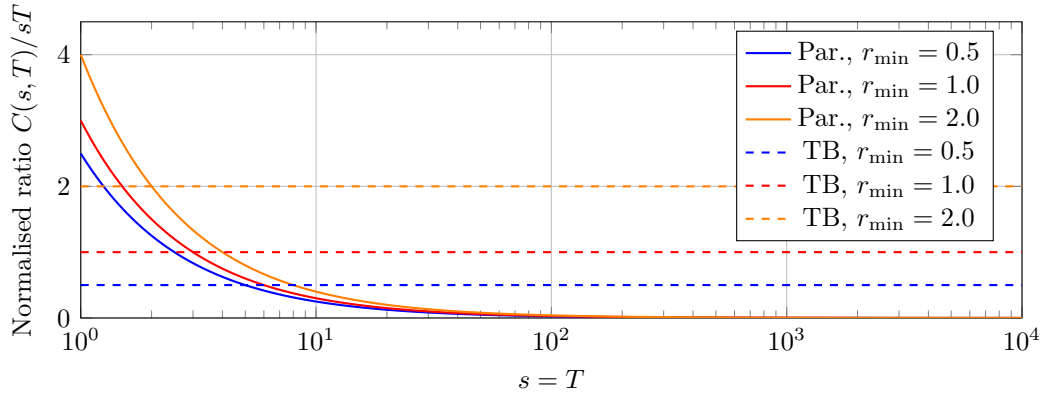
\begin{figure}[htbp]
		\centering
		\begin{tikzpicture}
			\begin{axis}[
				width=0.85\columnwidth,
				height=5.5cm,
				xmode=log,
				xlabel={$s = T$},
				ylabel={Normalised ratio $C(s,T)/sT$},
				xmin=1, xmax=10000,
				ymin=0, ymax=4.5,
				legend pos=north east,
				legend style={font=\small},
				grid=major,
				tick label style={font=\small},
				label style={font=\small},
				]
				\addplot[blue, thick, domain=1:10000, samples=200]
				{(0.5+2)/x};
				\addlegendentry{Par., $r_{\min}=0.5$}
				
				\addplot[red, thick, domain=1:10000, samples=200]
				{(1.0+2)/x};
				\addlegendentry{Par., $r_{\min}=1.0$}
				
				\addplot[orange, thick, domain=1:10000, samples=200]
				{(2.0+2)/x};
				\addlegendentry{Par., $r_{\min}=2.0$}
				
				\addplot[blue, dashed, thick, domain=1:10000]
				{0.5};
				\addlegendentry{TB, $r_{\min}=0.5$}
				
				\addplot[red, dashed, thick, domain=1:10000]
				{1.0};
				\addlegendentry{TB, $r_{\min}=1.0$}
				
				\addplot[orange, dashed, thick, domain=1:10000]
				{2.0};
				\addlegendentry{TB, $r_{\min}=2.0$}
			\end{axis}
		\end{tikzpicture}
		\caption{Normalised ratio $C(s,T)/sT$ along the joint diagonal $s=T$. Solid curves satisfy $(r_{\min}+2)/s \to 0$, illustrating the asymptotic behavior established by Theorem~\ref{thm:impossible}. Dashed lines remain fixed at $r_{\min}$, illustrating the linear-cost regime of Theorem~\ref{thm:linear}.}
		\label{fig:joint_asymptotic}
	\end{figure}
	
	\subsection{Non-Amplification Under Identity Replication}
	\label{ssec:eval_nonamplification}
	
	Under a throughput-bounded resource, influence is determined by channel capacity, not identity count. An adversary controlling $m$ channels and $s\in\{400,700,1000\}$ identities against $n=200$ honest validators achieves influence share $m/(m+n)$, independently of $s$. Since $s>m$ throughout, all three curves in Figure~\ref{fig:influence_share} overlap exactly: multiplying adversarial identities by $2.5$ yields zero influence gain, confirming the non-amplification consequence of Theorem~\ref{thm:linear}.
	
	\begin{figure}[htbp]
		\centering
		\includegraphics[width=0.72\columnwidth]{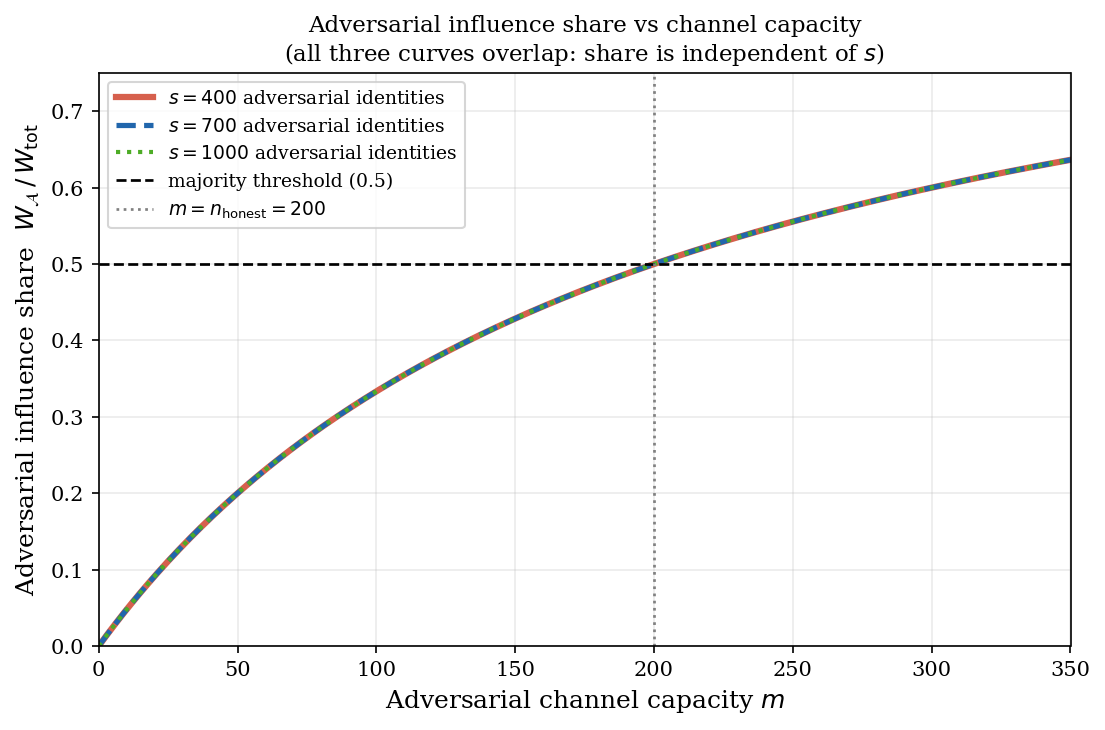}
		\caption{Adversarial influence share vs.\ channel capacity $m$, for $s\in\{400,700,1000\}$ and $n_{\rm honest}=200$. All curves overlap: influence share equals $m/(m+n)$, independent of $s$. Identity replication alone cannot increase influence; additional influence requires additional throughput channels (Theorem~\ref{thm:linear}).}
		\label{fig:influence_share}
	\end{figure}

	\section{Limitations and Open Problems}
	\label{sec:limitations}
	
	The structural separation established in this paper
	characterizes adversarial cost scaling under explicit
	modeling assumptions.
	Several limitations and open directions remain.
	
	\subsection{Enforcement of Throughput Constraints}
	\label{ssec:lim_enforcement}
	
	The framework assumes that the structural properties
	of Definition~\ref{def:throughput}---non-transferability,
	window-locality, and per-channel rate limits---hold
	in the deployment.
	In practice, mechanisms may only approximate these
	properties: device attestations can be rented or
	remotely operated, human participation can be
	outsourced, and advances in automation may effectively
	increase~$\tau$, weakening the per-channel rate
	constraint.
	Characterizing the degradation of the
	$C(s,T) = \Omega(sT)$ guarantee as a function of
	deployment-level fidelity remains an open problem.
	
	\subsection{Non-Additive Hybrid Compositions}
	\label{ssec:lim_hybrid}
	
	The current analysis establishes preservation of the
	linear lower bound only under the additive hybrid
	model of Section~\ref{sec:separation}, where total
	adversarial cost decomposes as
	$C(s,T) = C_p(s,T) + C_t(s,T)$ and a throughput-bounded
	component gates all influence units.
	Two open directions remain.
	First, non-additive compositions---where the two
	resource classes interact in ways not captured by
	additive decomposition---fall outside the present
	scope; it is not yet known whether the linear lower
	bound survives under such interactions.
	Second, partial gating arrangements, in which a
	throughput-bounded component constrains only a
	fraction~$\beta$ of influence units, may yield
	intermediate scaling regimes of the form
	$C(s,T) = \Omega(\beta \cdot sT)$, but a general
	characterization remains open.
	
	\subsection{Economic Coordination and Identity Markets}
	\label{ssec:lim_markets}
	
	The model abstracts away economic coordination
	dynamics.
	In practice, delegation markets, identity leasing,
	credential outsourcing, and secondary markets for
	rate-limited participation rights may alter the
	effective adversarial cost structure even when the
	underlying structural resource properties remain
	unchanged.
	An adversary with access to liquid identity markets
	may, for instance, effectively circumvent
	non-transferability constraints by acquiring
	participation slots through market intermediaries
	rather than through direct channel ownership.
	Integrating equilibrium analysis---capturing
	strategic outsourcing, adaptive pricing, and
	rational credential supply---with the structural
	framework developed here is an important direction
	for future research.
	
	\subsection{Scope of Assumption~\ref{asm:coordination}}
	\label{ssec:lim_assumption}
	
	Assumption~\ref{asm:coordination} provides a
	sufficient condition for the sublinear cost result of
	Theorem~\ref{thm:impossible}.
	Its purpose is not to model all conceivable
	coordination structures, but to isolate the
	contribution of resource reusability and
	transferability from coordination-specific effects.
	It reflects the coordination overhead observed in
	deployed infrastructures such as mining pools and
	validator delegation services, both of which exhibit
	coordination costs consistent with
	$h(s,T)=O(s+T)$
	\cite{rosenfeld2011analysis,saleh2021blockchain}.
	Whether the sublinear cost result extends to weaker
	coordination assumptions---in which overhead grows
	faster than $O(s + T)$ but remains $o(sT)$---is
	an open theoretical problem.
	
	\subsection{Window Granularity and Temporal Modeling}
	\label{ssec:lim_windows}
	
	The model partitions time into discrete windows of
	fixed duration.
	Smaller windows increase renewal frequency under
	throughput-bounded resources; larger windows may
	weaken window-local constraints.
	A window-independent formulation expressing the
	separation in terms of wall-clock time remains an
	open direction.
	
	\subsection{Scope of the Framework}
	\label{ssec:lim_scope}
	
	The framework is structural and protocol-agnostic:
	it characterizes the resource-level preconditions
	that any mechanism must satisfy, and does not replace
	deployment-specific security analyses, probabilistic
	consensus proofs, or economic equilibrium models.
	Characterizing intermediate scaling regimes---where
	structural properties are only partially satisfied---in
	full generality remains an open theoretical problem.
	The model is also resource-centric and static: it
	characterizes the cost $C(s,T)$ induced by a fixed
	resource-mechanism pair, but does not model settings
	in which the pair $(\mathcal{R}, f)$ itself evolves
	over time---for example, through protocol updates
	made in response to observed adversarial behavior.
	Extending the framework to such dynamic
	resource-mechanism pairs is an important direction
	for future work.
	\section{Conclusion}
	\label{sec:conclusion}
	
	Binding influence to a scarce resource prevents cheap
	identity creation from amplifying adversarial power---
	but scarcity alone is insufficient.
	The structural properties of the resource-mechanism pair
	determine whether influence cost is amortizable over time
	or necessarily recurring.
	
	This paper provides an axiomatic resource taxonomy
	connecting structural properties to adversarial cost
	scaling $C(s,T)$.
	The negative half of the taxonomy (Theorem~\ref{thm:impossible})
	establishes that any resource-mechanism pair satisfying
	divisibility, additivity of influence, temporal reusability,
	and identity transferability admits influence amortization:
	$C(s,T) = o(sT)$, with marginal cost $\Delta(s,T) = o(T)$.
	The positive half (Theorem~\ref{thm:linear}) identifies
	throughput-bounded, non-transferable, window-local resources
	as a class that provably enforces $C(s,T) = \Omega(sT)$,
	with marginal cost $\Delta(s,T) = \Omega(T)$.
	Theorems~\ref{thm:impossible} and~\ref{thm:linear}
	together establish a sharp asymptotic separation
	between the two resource classes
	(Theorem~\ref{thm:separation}).
	
	The Resource Substitution Theorem
	(Theorem~\ref{thm:substitution}) translates this
	separation into a design rule: enforcing linear
	cost of influence concentration requires grounding
	participation in a resource that violates at least
	one parallelizability property.
	Within the scope of Definition~\ref{def:parallelizable},
	no protocol rule can enforce this boundary without a
	structural change in the underlying resource.
	
	More broadly, these results suggest that
	decentralization cannot be analyzed solely at the
	protocol layer.
	The structural properties of the underlying resource
	constitute a first-order design choice that shapes
	the long-term economics of influence concentration.
	
	Future directions include the formal characterization
	of intermediate scaling regimes under partial
	violations of parallelizability, equilibrium analysis
	under adaptive incentives and identity markets, and
	the design of practical mechanisms that enforce
	throughput-boundedness while preserving
	permissionlessness and accessibility.

	\bibliographystyle{plain}
	\bibliography{references}
	
\end{document}